\documentclass{article}
\usepackage{amsmath}
\usepackage{amssymb}
\usepackage{geometry}           
\usepackage{mathrsfs}
\usepackage{graphicx}
\newcommand{\E}{\mathbb{E}}
\renewcommand{\P}{\mathbb{P}}
\newtheorem{theorem}{Theorem}
\newtheorem{prop}{Proposition}
\newtheorem{definition}{Definition}
\newtheorem{remark}{Remark}
\newtheorem{proof}{Proof}
\usepackage{natbib}
\title{Multiple testing using uniform filtering of ordered
  $p$-values}

\author{
Zhiwen Jiang\\
F. Hoffmann-La Roche AG and \'Ecole polytechnique f\'ed\'erale de
Lausanne (EPFL)\\
Basel, Switzerland \and
Stephan Morgenthaler\\
\'Ecole polytechnique f\'ed\'erale de Lausanne (EPFL) and REM
Analytics S.A.\\ Lausanne,
Switzerland\\
email: stephan.morgenthaler@epfl.ch}

\begin{document}

\maketitle

\begin{abstract}
  We investigate the multiplicity model with $m$ values of some
  test statistic independently drawn from a mixture of no effect
  (null) and positive effect (alternative), where we seek to identify
  the alternative test results with a controlled error rate.  We are
  interested in the case where the alternatives are rare. A number of
  multiple testing procedures filter the set of ordered p-values in
  order to eliminate the nulls. Such an approach can only work if the
  $p$-values originating from the alternatives form one or several
  identifiable clusters. The Benjamini and Hochberg (BH) method, for
  example, assumes that this cluster occurs in a small interval
  $(0,\Delta)$ and filters out all or most of the ordered $p$-values
  $p_{(r)}$ above a linear threshold $s \times r$. In repeated
  applications this filter controls the false discovery rate via the
  slope $s$. We propose a new adaptive filter that deletes the
  $p$-values from regions of uniform distribution. In cases where a
  single cluster remains, the $p$-values in an interval are declared
  alternatives, with the mid-point and the length of the interval
  chosen by controlling the data-dependent FDR at a desired level.
\end{abstract}

\textbf{keywords}:
  False discovery rate (FDR),   
  positive FDR,
  local FDR,  
  filtering of $p$-values,  
  heavy-tailed distribution,  
  mode estimation.


\section{Introduction}
The weighing of empirical evidence is important in many
disciplines and has gained new interest when combined with multiplicity. The classical
statistical tests are all designed for distributions linked to Gaussian errors
and this remains true in the multiple testing literature. 
The simplest random model for multiple tests is thus the Gaussian shift model
\begin{equation}
\label{Eqn:GaussianMix}
X_i \overset{\text{i.i.d.}}{\sim}
(1-\varepsilon)N(0, 1)+\varepsilon N(\mu, 1) \text{ for } i=1,\ldots,m \,,
\end{equation}
with $0\leq \varepsilon$ being the probability of drawing a true
alternative and $0< \mu$ the effect size. This occurs when the test
statistic is equal to a standardized average of repeated measures with
known variance.  Multiple testing refers to simultaneously considering
the family of null hypotheses $\{H_{0,i}:X_i\sim N(0,
1)\}_{i=1}^m$. Well known approaches include the Bonferroni-corrected
individual tests and the BH filter described by
\cite{benjamini1995controlling}. The first rejects only if the
$p$-value of an individual test is below $\alpha/m$, while the second
rejects the ordered $p$-values $p_{(r)}$ for $r<k$ if $k$ is the
smallest rank with the property that $p_{ (r)}>r\alpha/m$ for all
$r\geq k$. The first method ensures that the family-wise error rate
(FWER) is below $\alpha$, while the second has a false discovery rate
bounded by $\alpha$.  Besides identifying which of the $H_{0,i}$ to
reject, we might also be interested in the global test, which for
model (\ref{Eqn:GaussianMix}) is
\begin{equation}
\label{Eqn:GaussianMixGlobal}
H_{0}\ :\ \varepsilon=0 \text{ against }H_{1}\ :\ \varepsilon>0 \,.
\end{equation}
With the mixture model, probabilities such as
$P\left(H_{0,i} \text{ is true }\vert\text{ test rejects } H_{0,i}
\right)$ can be considered. This Bayesian point of view has been
investigated among others by \cite{efron2001empirical} or
\cite{efron2002empirical}.

A more general two group model is as follows:
\begin{equation}
\label{Eqn:MultipleTestingCh4}
 H_{0,i}:\ X_i \sim F_0 \quad \text{against} \quad H_{1,i}:\ X_i \sim F_1\,,  \quad   i=1, \ldots, m,   
\end{equation} 
Let $H^m=(H_1, \ldots, H_m)$ be the indicator variables where $H_i=0$
if and only if the $i$-th null hypothesis is true. It follows that
$H_i \overset{\text{i.i.d.}}{\sim} \text{Bernoulli}(\varepsilon)$.  We
will work in the case of continuous univariate distributions.  The
random $p$-value for the i-th test is
 \begin{equation}
 P_i=1-F_0(X_i)\,.
 \end{equation}
 If $X_i$ has distribution $F_0$, this random $p$-value is uniformly
 distributed in $(0,1)$.  If $X_i$ has distribution $F_1$, the
 distribution of $P_i$ is
$$F_{P}(t) = 1-F_1(F_0^{-1}(1-t))\text{ for }0<t<1  \,,$$  
with density function   
\begin{equation}
f_{P}(t) = \frac{f_1(F_0^{-1}(1-t))}{f_0(F_0^{-1}(1-t))}
= \exp\left(\mu z_{(1-t)}-\mu^2/2\right)\,.
\label{Eqn:DensityPvalue}
\end{equation}
The right-most expression applies to the Gaussian model
(\ref{Eqn:GaussianMix}) with effect size $\mu>0$ and the notation for
Gaussian quantiles $z_{1-t} = \Phi^{-1}(1-t)$.  The marginal density
of a randomly selected $p$-values under the mixture model is
$$\tilde f_{P}(t) 
=(1-\varepsilon) +\varepsilon\,\frac{f_1(F_0^{-1}(1-t))}{f_0(F_0^{-1}(1-t))}\text{ for }0<t<1,
$$
with distribution function $\tilde F_P(t)$.

In this article we investigate multiple testing procedures that work
with the complete set of $p$-values. The idea is to apply a filter
which deletes $p$-values in all regions where uniformity seems to
hold, thus making the search for regions containing true alternatives
easier. If a procedure classifies all test results with $p$-values in
the interval $[t\pm\Delta/2]$ as non-null, the true positive rate is
$\text{TPR}(t)=F_p(t+\Delta/2)-F_p(t+\Delta/2)$, while the false
positive rate equals $\Delta$. The larger the ratio between the two,
the better the procedure. In the limit, as $\Delta \to 0$ the ratio
converges to $f_P(t)$, which shows that one must look for modal values
of this density. Clusters of $p$-values are indicative for such modal
values.  The uniform filter will reveal the regions of clustering by
eliminating null $p$-values. If no such clusters emerge, it is due to
the second term in $\tilde f_{P}(t)$ not being sufficiently distinct
from the uniform first term. In such situations, no generic multiple
testing procedure can identify the alternatives, because they are not
detectable. In the two group model this occurs for small differences
between $F_0$ and $F_1$ or small values of $\varepsilon$ or both and
is an indication for underpowered studies.

The Cauchy case is considered in this paper as a radical counter
example to the Gaussian. The
shift model applies, if based on a single observation, the center of
symmetry is to be tested. For the standard Cauchy distribution $C(x)$
with the alternatives shifted by $\mu$, the density of the $p$-value is
\begin{align*} 
\tilde f_{P}(t)  
= (1-\varepsilon) +\varepsilon\, \left( {1+\frac{1}{\tan^2 (\pi t)}}\right) \left /  
\left(1+ \left( \frac{1}{\tan (\pi t)} - \mu\right)^2 \right) \right. .
\end{align*}
The corresponding density $f_P$ is bounded and with a small fraction
$\varepsilon$, the density $\varepsilon f_P$ is not easily
distinguished from the value of $1$ of the uniform density. In
addition, the mode of this density is at
\begin{equation}
\label{Eq:mode}
0 < p_c= \frac{1}{\pi}\left(\arctan \left(-\sqrt{1+\frac{\mu^2}{4}}-
    \frac{\mu}{2} \right)\right) +\frac{1}{2}\,.
\end{equation}
Other Cauchy alternatives could be considered. The t-test with samples
of size 2 has a standard Cauchy null distribution and a noncentral t
with one degree of freedom alternative. For testing a shift
alternative with several observations the median is an efficient test
statistic.

The goal of this paper is to propose an adaptive and broadly
applicable methodology for the multiplicity problem by relying on the
uniformity of all the $p$-values resulting from null hypotheses. This
will include the possibility of long-tailed distributions of the test
statistic. \cite{fan2019farmtest} also investigated the robust
multiple testing problem for correlated and long-tailed data, based on
an adaptive Huber covariance estimator and a factor-adjusted model.

\section{Uniform filtering of $p$-values}

In multiple testing, we want to know which hypotheses are most likely
false and how many of these should be rejected? As we remarked above,
this requires estimating the mode of $f_P(t)$, a problem can be
tackled via a density estimate of $\tilde f_P$. Some papers such as
\cite{efron2001empirical}, \cite{efron2004large},
\cite{genovese2004stochastic} and \cite{jin2007estimating} have made
contributions in this direction, but their methods are not adapted to
the range of cases we investigate here, because modes of $f_P$ may be
impossibly hard to spot in $\tilde f_P$.  In order to obtain an
accurate estimator of the mode we propose a method that reduces the
noise caused by the $p$-values resulting from the true null hypotheses
by deleting a fraction of the observed $p$-values.

\subsection{Fixed-length filters}

Let $ p^m = \left\{ p_1, \ldots, p_m \right\} $ be the complete set of
observed $p$-values. Our filter deletes the $p$-values guided by a
regularly spaced grid of interval midpoints. Suppose we plan to delete
$m_{\xi}= \lceil (1-\xi)m \rceil $ of the $p$-values in $p^m$, where
$\xi \in (0,1)$ is a tuning parameter. Consider the bins
$I_j = (c_j-1/(2m_\xi),c_j+1/(2m_\xi)]$ centered at
$c_j = (2j-1)/(2m_\xi)$ $(j=1,\ldots,m_\xi)$. The proposed filter
$\mathcal{T}$ uses the minimal distance between the
$p$-values and the bin centers $c_j$ as the deletion
criterion. Starting from $j=1$, it runs through $j=1,\ldots,m_\xi$ and
deletes for each $j$ among the remaining $p$-value the one closest to
$c_j$. It thus deletes exactly $m_\xi$ of the $p$-values in $p^m$. One
could also run the filter in the inverse direction, but this does not
make a big difference.

To heuristically investigate the properties of the filtering, consider
again the mixture model
$$ (1-\varepsilon) F_0(x) + \varepsilon F_1(x) $$
and delete $m_\xi$ values from $p^m$. Assume that $\xi > \varepsilon$
and that the density $f_P$ is Riemann integrable. Approximate $f_P$ by
a mixture of uniform distributions on the filter intervals of length
$1/m_\xi$. In any of the filter intervals $(l,u]$ containing
$p$-values from both null and alternative tests, the probability that
the $p$-value closest to an interval center is a true alternative is
then -- using the uniform approximation -- equal to the ratio of
$\varepsilon(F_P(u)-F_P(l))$ to
$(1-\varepsilon)(u-l)+ \varepsilon(F_p(u)-F_P(l))$. This can be
approximated by
$\big[\varepsilon f_P((u+l)/2)]/((1-\varepsilon) +\varepsilon
f_P((u+l)/2)\big]$.  This probability is equal to the expected number
of rejections of true alternatives. Summing over all intervals shows
that the expected number of false deletions due to the filter is
approximately equal to
\begin{equation}
I=\xi m \int_0^1 \frac{\varepsilon f_P(x)}{(1-\varepsilon) + \varepsilon f_P(x)} \,dx\,.
\label{FalseDels}
\end{equation}
Table \ref{tab:simulA} shows a comparison between the average of 10
simulations with the theoretical value of Eq. (\ref{FalseDels}) for
both Gaussian and Cauchy cases. The entries show that the formula
gives a useful approximation for large values of $m$.

\begin{table}
  \caption{\label{tab:simulA} The first five columns describe the
    mixture model for the test statistics with $F_1(x) = F_0(x-\mu)$
    and give the number $m$ of tests as well as the value of
    $\xi$. The sixth column is the number of $p$-values filtered
    out. The column true is the actual number of true alternatives,
    while the approximate expectation and the average over the
    simulations of the number of remaining alternative $p$-values are
    in the last two columns. The average is over 10 independent
    replications. The value of the integral was approximated by
    Riemann integration with 2000 equal-length intervals. The value in
    parenthesis is the standard error of the average.}
\centering
\begin{tabular}{cccc|c|cccc}
&&\multicolumn{4}{c}{~}&&\multicolumn{2}{l}{number of remaining alternatives}\\
  $F_0$&$\varepsilon$&$\xi$&$\mu$&$m$&$(1-\xi)m$&true&theoretical&simulated\\
  \hline
Gaussian&0.01&0.05&2&40000&38000&400&79.7&78.1 ($\pm 2.72$)  \\
Gaussian&0.01&0.05&3&40000&38000&400&204.3&202.8 ($\pm 3.0$) \\
Gaussian&0.01&0.05&5&40000&38000&400&374.6&373.3 ($\pm 1.1$) \\
Gaussian&0.01&0.01&5&40000&39600&400&373.6&373.0 ($\pm 4.9$) \\
Gaussian&0.01&0.005&5&40000&39800&400&--&199.9 ($\pm 0.3$) \\
Cauchy&0.01&0.05&10&40000&38000&400&130.6&124.4 ($\pm 3.5$) \\
Cauchy&0.01&0.05&20&40000&38000&400&229.4&230.2 ($\pm 3.1$) \\
Cauchy&0.01&0.05&40&40000&38000&400&307.4&306.0 ($\pm 1.8$)       
\end{tabular}

\end{table}

Before filtering the fraction of true alternatives is
$\varepsilon$. After the application of the filter, the fraction in
the remaining $p$-values from true alternatives is approximately equal
to $(\varepsilon m - (1-\xi) m I)/(\xi m)$, where $I$ is the value of
the integral in Eq. (\ref{FalseDels}). The maximal possible fraction
is $\varepsilon/\xi$, which occurs if $I$ is very close to 0 and the
filter only deletes $p$-values originating from true nulls. When
$F_0=F_1$ and $I=\varepsilon$, the fraction of true alternatives among
the remaining $p$-values remains at the original value of
$\varepsilon$.  The filter is intended to enrich the true alternatives
among the remaining $p$-values. The formula shows that the enrichment
works best with small values of $\xi$.

\subsection{Large $m$ and small $\varepsilon$}

Table \ref{tab:simulA} also shows the effect of the shift $\mu$ on the
number of false deletions of true alternatives. The larger the shift,
the easier it is to detect the alternatives. The shift of 5 standard
errors with a Gaussian test statistic is examined by choosing three
small values of $\xi$. The numbers show that the uniform filter quite
easily finds the most relevant range of $p$-values. In 9 of the 10
simulations with the smallest value of $\xi$, all the remaining
$p$-values were alternatives. What would happen, if we chose a larger
number of tests? How large a $\mu$ would we need to reach similar
certainty. If $\varepsilon = \varepsilon_m$ goes to zero as the number
of tests increases, the numerator of the integrand in
(\ref{FalseDels}) goes to zero and the uniform filtering will not
work. In order to make the problem interesting, the mode in the
$p$-value density also has to grow as $m$ increases.

Suppose the mode of $f_{P,m}$ is at location $M_m$.  In order for the
filtering to work, the clustering must occur on the scale defined by
$1/m$ and $\varepsilon_n f_{P,m}(M_m+ 1/m)$ must grow with larger
$m$. Let the fraction $\varepsilon_m$ of true alternatives be a
decreasing fraction of $m$, for example, by choosing
$\varepsilon_m = m^{-\gamma}$ with $0.5<\gamma<1$. In the Gaussian
case, the peak will become more pronounced if $\mu_m$ grows with
$m$. The density $f_{P,m}(x)= \exp(-(z_x-\mu_m)^2/2)/\exp(-z_x^2/2)$
has a singularity at $M_m=0$. For small $x$,
$z_x\approx \sqrt{-2\log(x)}$ and thus
$f_{P,m}(0+1/m)\approx \exp(\sqrt{2\log(m)}\mu_m-\mu_m^2/2)$. Both
terms are multiples of $\log(m)$, if $\mu_m= C \sqrt{\log(m)}$ is
proportional to $\sqrt{\log(m)}$. Multiplying this by $m^\gamma$ leads
to
$m^{\sqrt{2}C-C^2/2-\gamma}$. The smaller root of this quadratic is
$C_\text{low} = \sqrt{2*(1-\sqrt{(1-\gamma)} )}$. Any $C > C_\text{low}$
  will lead to a dense cluster of alternatives that can be detected by
  uniform filtering. See also \cite{ingster1997},
\cite{donoho2004higher}.

\subsection{Exploring the $p$-values from a multiple test}

The fixed-length filter offers a way to explore the set of $p$-values
from a multiple test. To do so, divide the interval $(0,1)$ into $m$
equal length bins and count then number of $p$-values in each bin. If
all $m$ $p$-values were uniformly distributed, the expectation for all
counts would be 1 and the counts would approximately follow a Poisson
distribution with expectation 1.  
It is easiest to plot the counts in the order of the bins, starting
from $p=0$. Fig. \ref{Fig:Barplot} shows four examples with different
distributions and different values of the shift $\mu$. In the upper
left corner is a case of $m=200$ test based on a test statistic with
Gaussian distribution and a shift or effect of $\mu=1$. There clearly
is no cluster in the range of $p$-values below 0.2. If the shift is
increases to 4, a clear peak of nine $p$-values appears in the first
bin. The lower row in Fig. \ref{Fig:Barplot} shows longer tailed test
statistics. In the Student's $t_3$ case, a hint of a peak is in the
second bin, which does indeed contains two $p$-values generated by
alternatives, whereas the second Cauchy-distributed case shows a very
clear peak of five bins with $p$-values ranging from 0.0077 to
0.0082.

\begin{figure}
\centering
\includegraphics[width=0.95\textwidth]{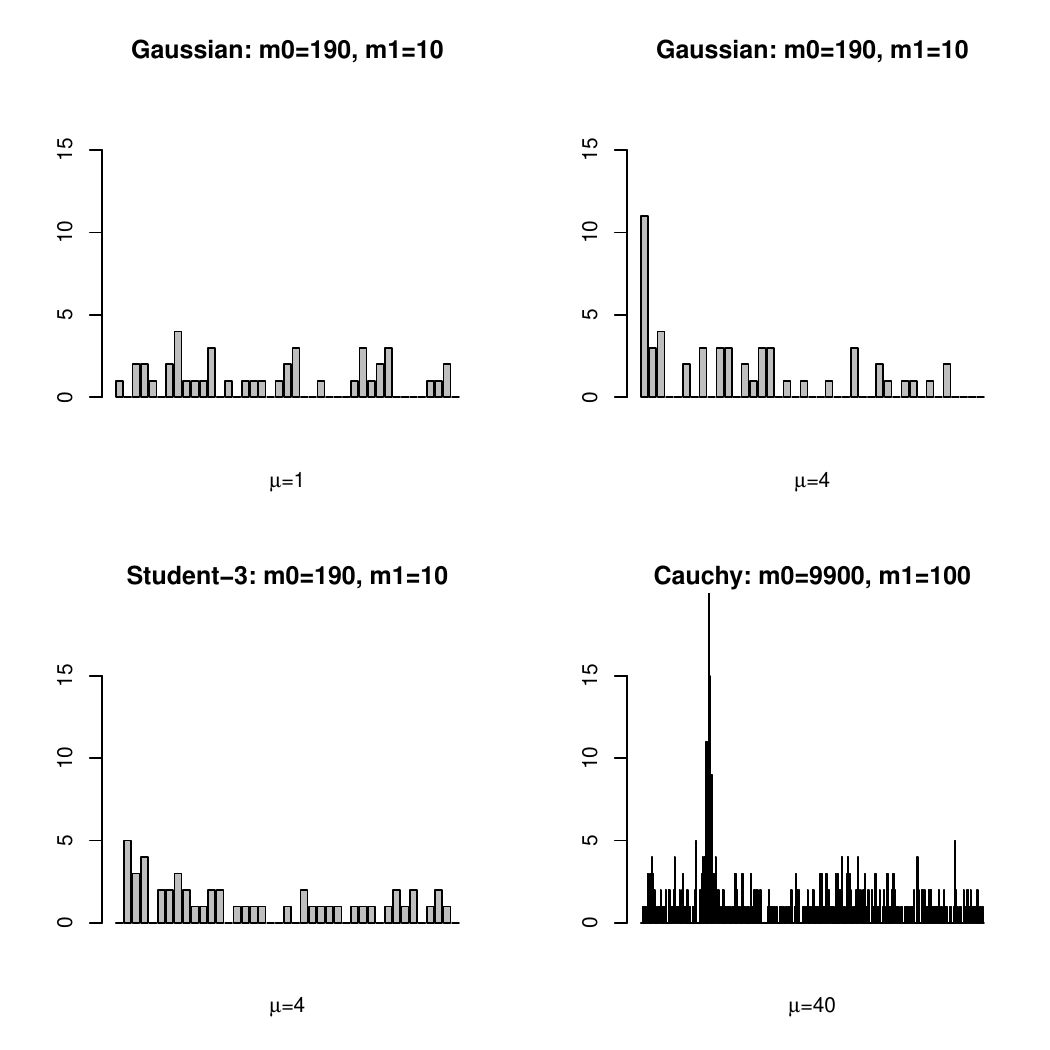}
\caption{\label{Fig:Barplot} The plots show the counts in the bins of
  width $1/m$ up to the value of 0.2. The plots on the left show
  examples without clusters. The counts of null $p$-values in the
  filter intervals form a Poisson process, whose maxima grow very
  slowly with $(1-\varepsilon)m$ and are around 4 to 5 in all cases
  shown. The plots on the right exhibit clear clusters.}
\end{figure}

\subsection{Behavior for long-tailed distributions}

In contrast to the Gaussian case, the following results covers a
long-tailed distribution and with a filter that is easier to
analyze. The ``random filter'' also uses a regular grid with intervals
of length $1/m_\xi$. It operates by deleting one randomly chosen
$p$-value from each non-empty interval. 
A random filter deletes less than $m_\xi$ $p$-values because of empty
intervals. Since the gaps between the ordered uniform $p$-values have
a beta$(1,(1-\varepsilon)m)$ distribution, the probability of a gap
larger than $2/m_\xi$ is equal to
$(1-2/m_\xi)^{(1-\varepsilon)m}\approx
\exp(-2(1-\varepsilon)/(1-\xi))$. Gaps of this size imply at least one
empty interval.

The result shows that the
uniform filtering works as required, but the range of detectability
changes. The next result shows that the randomised filtering procedure
essentially preserves the true alternatives in the limit as $m\to\infty$. The
tight clustering of the $p$-values from the alternatives are
responsible for this result.

\begin{theorem}[Asymptotic filtering] 
We consider the standard Cauchy mixture model,
$$ (1-\varepsilon_m) C(x) + \varepsilon_m C(x-\mu_m) \,,$$
where $ \varepsilon_m = m^{-\gamma} $ and $ \mu_m = m^{r} $.  If the
parameters $(\gamma, r)$ satisfy
$$
0 < \gamma\text{ and } r > 1- \frac{\gamma}{2}\,,
$$
it follows that if we use the ``random filter'' with parameter $\xi$, the expected fraction
of the number $\text{FE}_\xi$ of falsely deleted alternatives among all true
alternatives $\varepsilon_mm$ converges to zero, that is,
$$
\label{Eqn:EFEvanishing}
\frac{\E( \text{FE}_\xi) }{\varepsilon_m m}  \longrightarrow 0,  \quad  m \to \infty, 
$$
\end{theorem}
\begin{proof}
The expected value of the false exclusions is
$$
\E( \text{FE}_\xi) = \E \left(\sum_{i:\, H_i =1} \textbf{1} \left \{    P_i \in \mathscr{D}     \right \}\right) \,,
$$
where $\textbf{1}$ denotes an indicator function and $\mathscr{D}$ the
random set of deleted $p$-values. 
With the filter intervals $\left\{ I_j,\; j=1, \ldots,  m_{\xi}
\right\}$ on $[0,1]$ and their mid-points $c_j$, we obtain
\begin{equation*}
\begin{aligned} 
\E(\text{FE}_\xi) & =  \sum_{i:\, H_i =1}
\sum_{j=1}^{m_{\xi}} \P \left( P_i \in ( \mathscr{D} \cap I_j )
\right)    = \varepsilon_m m  \sum_{j=1}^{m_{\xi}}
\frac{ \varepsilon_m f_P (c_j) }{  (1-\varepsilon_m) + \varepsilon_m f_P
  (c_j) } + O(1/m)  \\
& = \varepsilon_m m  m_{\xi}  \int_0^1
\frac{\varepsilon_m f_P (t)}{(1-\varepsilon_m) + \varepsilon_m f_P (t)}
\,\text{d}t + O(1/m)  =  \varepsilon_m \varepsilon_m m  m_{\xi}
\int_0^1 \frac{1 }{(1-\varepsilon_m) \frac{1}{f_P(t)} + \varepsilon_m  }
\,\text{d}t  .   \\
\end{aligned}
\end{equation*}

Since  the function $1/ f_P$ is bounded on $[0,1]$ and has a single peak at  
$$\frac{1}{\pi}\left(\arctan
  \left(-\sqrt{1+\frac{\mu^2}{4}}-\frac{\mu}{2} \right)\right) +1 \geq
\frac{1}{2}\,,$$ 
we have the integral over $[0,1]$ upper bounded by twice the integral
over the right half. Thus, 

\begin{equation*}
\begin{aligned}
\E(\text{FE}_\xi)
& \leq   2 \varepsilon_m m_{\xi} \varepsilon_m m
\int_{\frac{1}{2}}^1 \frac{1 }{(1-\varepsilon_m)  \frac{1+\left( \tan
      (\pi/2 - \pi t) - \mu \right)^2}{1+ \tan^2 (\pi/2 - \pi t) }+
  \varepsilon_m   } \,\text{d}t    \\
& =   \frac{2 \varepsilon_m^2 m_{\xi} m}{\pi }  \int_{0}^{\frac{\pi}{2}}
\frac{1}{(1-\varepsilon_m)  \frac{1+\left( \tan x + \mu \right)^2}{1+
    \tan^2 x }+ \varepsilon_m   }   \,\text{d}x     \\
& \leq   \frac{2 \varepsilon_m^2 m_{\xi}  m}{\pi }
\int_{0}^{\frac{\pi}{2}}    \frac{1}{(1-\varepsilon_m)  (1+ \mu^2)
  \cos^2 x  + \varepsilon_m   }   \,\text{d}x     \\
& =    \frac{2 \varepsilon_m^2 m_{\xi} m}{\pi }
\int_{0}^{\frac{\pi}{2}}    \frac{\sec^2 x}{(1-\varepsilon_m)  (1+
  \mu^2)   + \varepsilon_m (1+\tan^2 x)  }  \, \text{d}x     \\
& = \frac{2 \varepsilon_m^2 m_{\xi} m }{\pi }   \int_{0}^{\infty}
\frac{1}{(1-\varepsilon_m)  (1+ \mu^2)   + \varepsilon_m + \varepsilon_m y^2
}   \, \text{d}y     \\
& =  \frac{2 \varepsilon_m^2 m_{\xi} m}{\pi }   \frac{1}{(1-\varepsilon_m)
  (1+ \mu^2)   + \varepsilon_m }  \int_{0}^{\infty}   \frac{1}{  1+
  \left( \frac{\sqrt{\varepsilon_m} y }{  \sqrt{ (1-\varepsilon_m)  (1+
        \mu^2)   + \varepsilon_m  }} \right)^2 }  \, \text{d}y     \\
& =  \varepsilon_m^2 m_{\xi}  m  \frac{1}{\sqrt{\varepsilon_m} \sqrt{
    (1-\varepsilon_m)  (1+ \mu^2)   + \varepsilon_m  } }\,. 
\end{aligned}
\end{equation*}

Recall that we consider the parametrisation 
$$ \varepsilon_m = m^{-\gamma} ,\quad  \mu_m = m^{r} \,.$$
Therefore, the expected proportion satisfies
$$
\frac{\E(\text{FE}_\xi) }{\varepsilon_m m}  
\leq  \frac{\sqrt{\varepsilon_m} m_{\xi} }{ \sqrt{ (1-\varepsilon_m)  (1+ \mu^2)   + \varepsilon_m  } }
= \frac{1-\xi}{ m^{r+\gamma /2-1} (1+ o(1) ) }
\longrightarrow 0  
$$
as  $m \to \infty,$ if 
$$
r >  1- \frac{\gamma}{2}. 
$$

\end{proof}
  

\begin{theorem}
  Consider the Cauchy mixture model for the test statistics with
  $ \varepsilon_m = m^{-\gamma} $ for $0.5 < \gamma <1$ and
  $ \mu_m = m^{r} $. Let $P_1,\ldots,P_m$ be an i.i.d. sample of
  $p$-values from this model and apply the fixed-length filter
  deleting $m_\xi = \lceil(1-\xi)m\rceil$ of the $p$-values. If
$$
r > 1- \frac{\gamma}{2}\,,
$$
it follows that
$$
\label{Eqn:EFExivanishing}
\frac{\E\, (\text{FE}) }{m\varepsilon_m}  \longrightarrow 0,  \quad  m \to \infty\,.
$$
\end{theorem}

\begin{proof}
The sequence of deleted $p$-values, $\left\{ P_{j}^{\xi} ,\; j=1,
  \ldots,  m_\xi  \right\}$
forms a partition on $(0,1]$, which we denote by 
$$ I_j^* = \left( P_{(j-1)}^{\xi},\; P_{(j)}^{\xi} \right] ,   \quad j=1, \ldots,  m_\xi+1,$$
with $P_{(0)}^{\xi}=0$ and $P_{( m_\xi +1)}^{\xi}=1$. 
Let $\left\{ c_j^*,\; j=1, \ldots,  m_\xi +1  \right\}$ be the mid-points.  
It follows that
\begin{equation*}
\begin{aligned}
\frac{\E\, \big | \mathrm{FE}_{\xi}  \big | }{\varepsilon_m m}
& = \frac{1}{\varepsilon_m m} \E \sum_{i:\, H_i =1} \textbf{1} \left \{
  P_i \in \mathscr{D}     \right \}   \\
&=  \frac{1}{\varepsilon_m m}  \sum_{i:\, H_i =1}  \sum_{j=1}^{ m_\xi +1
} \P \left( P_i \in ( \mathscr{D} \cap I_j^* )    \right)     \\
& =   \sum_{j=1}^{ m_\xi+1} \frac{ \varepsilon_m f_P (c^*_j) }{
  (1-\varepsilon_m) + \varepsilon_m f_P (c^*_j) } + O(1/m)  \\
& =   ( m_\xi +1 ) \int_0^1 \frac{\varepsilon_m f_P
  (t)}{(1-\varepsilon_m) + \varepsilon_m f_P (t)} \,\text{d}t + O(1/m)  \\
&=  \varepsilon_m   ( m_\xi +1 )  \int_0^1 \frac{1 }{(1-\varepsilon_m)
  \frac{1}{f_P(t)} + \varepsilon_m  }  \,\text{d}t     \\
& = \frac{1-\xi}{ m^{r+\gamma /2-1} (1+ o(1) ) }
\longrightarrow 0, \quad m \longrightarrow \infty, 
\end{aligned}
\end{equation*}
if $r > 1- \frac{\gamma}{2}\,, $  which is the same asymptotic boundary as (\ref{Eqn:EFEvanishing}).
\end{proof}
'
\section{A multiple test based on uniform filtering}

In this section, we show how to formalize the filter approach to
obtain a multiple testing procedure in situations where a unique mode of the
alternative $p$-value density exists. Our method uses the remaining
$p$-values after filtering to estimate the mode. We then
return to the original sequence of $p$-values and determine the
rejection region surrounding the mode in such a manner
that the estimated false discovery rate is below a chosen value.
Fig. \ref{fig:CauchyFilter} shows an example.

\begin{figure}
    \centering
    \includegraphics[width=0.49\textwidth]{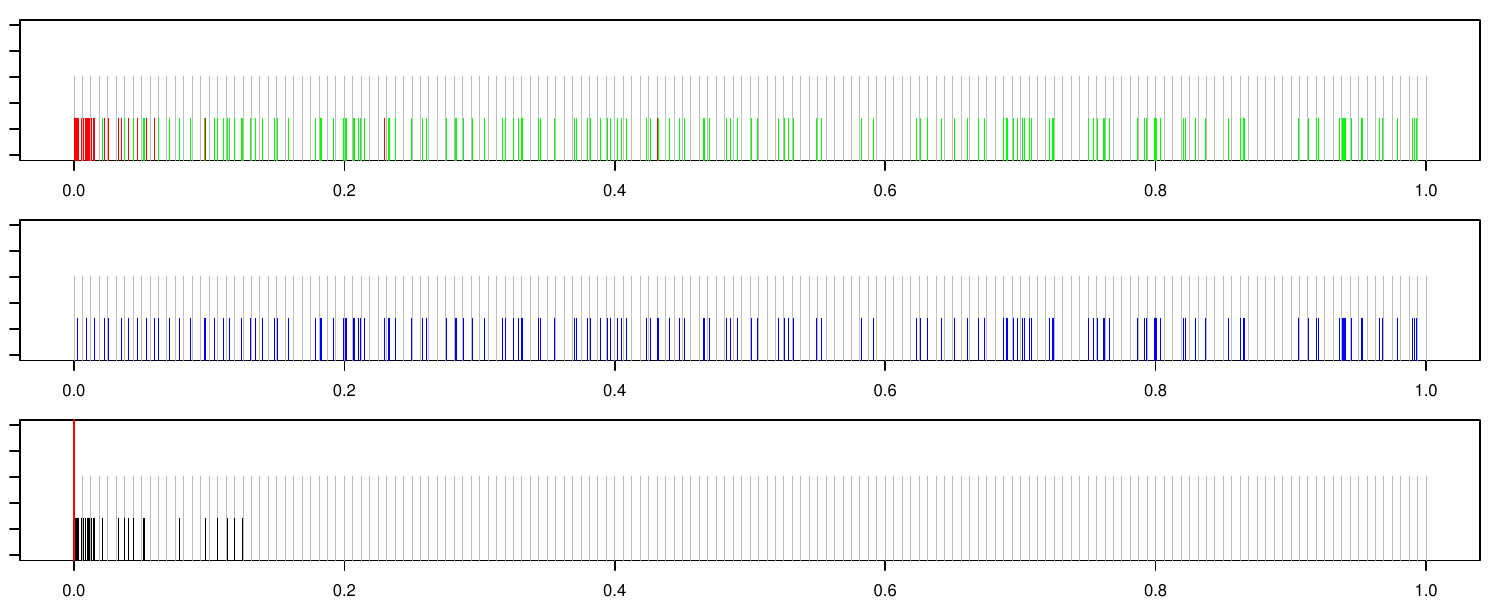}
    \includegraphics[width=0.49\textwidth]{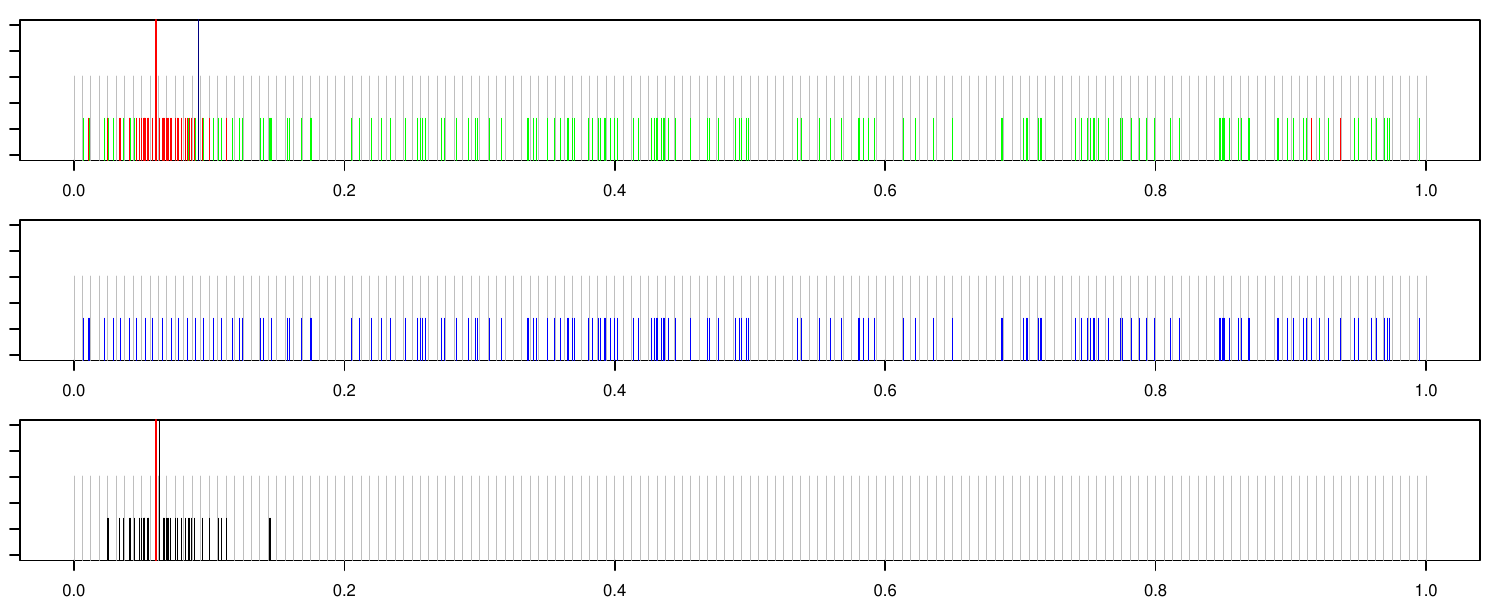}
    
    \caption{\label{fig:CauchyFilter} The effect of filtering and the
      identification of the mode. The plots in the top line show the
      sample of $p$-values. The red values correspond to the true
      alternatives. The middle lines show the deleted $p$-values and
      the bottom line are the remaining ones. The vertical lines
      indicate the estimated location of the mode. }
    
\end{figure}

\subsection{Convergence of the filtered mode as $m\to\infty$}
Let the unique mode be
$$
 \vartheta =  \underset{t}{\arg \max}\; \tilde f_P(t) =
 \underset{t}{\arg \max}\;  \left( (1-\varepsilon) + \varepsilon f_P
   (t) \right) =  \underset{t}{\arg \max} \;  f_P(t)
$$
and let $\{p^{*}_1, \ldots,p^{*}_{\lfloor \xi m\rfloor}\}$ be the
remaining $p$-values after the use of a fixed-length filter.

Assume that the cumulative distribution and the density of the
$p$-values under the alternative, namely $F_P$ and $f_P$, satisfy 
\begin{enumerate}
\item $F_P$ is not concave,  
\item $\lim_{t\to 0^{+}}\frac{\text{d} F_P(t)}{\text{d} t} = 1$,  
\item $\lim_{t\to 0^{+}}\frac{\text{d} f_P(t)}{\text{d} t} =1$,
\item $f_P$ is unimodal,
\item $f_P(t)$ is uniformly continuous in $t$.
\end{enumerate}

The filtered $p$-values $\mathscr{F}^\xi  = \left\{ P_1^{*}, \ldots, P_{\lfloor \xi m\rfloor}^{*} \right\}$
have distribution function $F_P^{*}$ and density function $f_P^{*}\,.$ 
We also assume that the density $f_{P}^{*} (t)$ has a unique mode $\vartheta^{\xi}$ defined by
$$
f_{P}^{*} (\vartheta^{\xi}) = \max_{0 \leq t \leq 1}\, f_{P}^{*} (t)\,.
$$

With a bandwidth $h$ and the kernel
function $K$, the filtered density estimate becomes 
$$
\hat f_P^{(\xi,h)}(t) = \frac{1}{h\lfloor \xi m\rfloor}
\sum_{j=1}^{\lfloor \xi m\rfloor} K\left( \frac{t-p^{*}_j}{h}
\right)\,, 
$$
of which the sample mode will be shown to converge to the true mode.

\begin{definition}
The random variable $\hat\vartheta^{\xi,h}$ such that 
$$
{\hat f_P}^{(\xi,h)} (\hat\vartheta^{\xi,h}) = \max_{0 \leq t \leq 1}\, {f_P}^{(\xi,h)}  (t)
$$
is called the sample mode.
\end{definition}

In kernel density estimation, the bandwidth $h$ balances the bias and the variance of the estimator, which implies that the value of $h$ has to be adapted to the value of $m$.  In the following, we consider the situation where $m\to\infty$ with a fixed value of $\xi$ and we assume $h=h(\xi,m)$ satisfies 
\begin{equation}
\label{eqn:BandWidth}
\lim_{\lfloor \xi m\rfloor \to \infty}\, h(\xi,m) = 0, \quad   \lim_{\lfloor \xi m\rfloor \to \infty}\, \lfloor \xi m\rfloor h^2(\xi,m)  =  \infty .
\end{equation}

\begin{theorem}
\label{thm:ModeConvergence}
Suppose $\hat\vartheta^{(\xi,h)}$ is the sample mode and $h=h(\xi,m)$ is a bandwidth satisfying (\ref{eqn:BandWidth}).  
Then $\hat\vartheta^{(\xi,h)} \to \vartheta$ in probability.

\end{theorem}

We use the following propositions to prove the theorem. 

\begin{prop}[Unimodal distribution]
\label{prop:Unimodal}
For a uniformly continuous and unimodal density on the interval $(0,1)$ with mode $\vartheta$, it follows that, for any $\delta >0$ there exists an $\delta ' >0$ such that, for  $0<t<1,$ 
\begin{equation*}
\left |  \vartheta - t  \right | \geq \delta \Longrightarrow  \left | f(\vartheta) - f(t)  \right |  \geq \delta ' \,.
\end{equation*} 
\end{prop}

\begin{prop}
\label{Prop:ConsistentMode}
Let  $\vartheta^{\xi}$ be the unique mode based on the distribution of  the filtered $p$-values $ \mathscr{F}^\xi$ and $\vartheta$ the mode with respect to the distribution of $p$-values under the alternatives. 
Then 
 $\vartheta^{\xi}  \to  \vartheta$ in probability, as $\lfloor \xi m\rfloor \to \infty.$ 
\end{prop}

\begin{proof}
Proposition \ref{prop:Unimodal} is easy to obtain.  
The proof of  Proposition \ref{Prop:ConsistentMode} follows from Hoeffding's inequality applied to 
$$ \frac{1}{\varepsilon m} \sum_{i: \, H_i=1} \textbf{1} \left\{ P_i \in \mathscr{F}^\xi \right\}\,,$$ 
where $\textbf{1} \left\{ P_i \in \mathscr{F}^\xi \right\} $ are bounded random variables. 
\end{proof}

The proof of Theorem \ref{thm:ModeConvergence} is as follows.
\begin{proof}
Consider the density function  $f_{P}^{*} (t)$ and the true mode $\vartheta^{\xi}$ based on the distribution of $ \mathscr{F}^\xi \,.$
Since $\hat\vartheta^{(\xi,h)}  $ is the sample mode derived from the kernel density estimator after filtration, we prove that 
$$
\hat\vartheta^{(\xi,h)}  \to \vartheta^{\xi},  \quad   \lfloor \xi m\rfloor \to \infty\,.
$$

As we assumed that $f_{P}^{*} (t)$ is uniformly continuous and has a unique mode $\vartheta^{\xi}\,,$ it follows that the Proposition \ref{prop:Unimodal} holds. 
Therefore, it is enough to prove the convergence of $f_{P}^{*} (\hat\vartheta^{(\xi,h)}  )$ in probability, that is,
$$
f_{P}^{*} (\hat\vartheta^{(\xi,h)} )  \overset{p}{\longrightarrow} f_{P}^{*} (\vartheta^{\xi} ),   \quad \lfloor \xi m\rfloor \to \infty\,.
$$
In order to derive the convergence of the estimated density function, we investigate the characteristic function of the sample.
Let $ \left\{ \varphi_j \right\}_{j =1}^{\infty}$  be the sequence of  sample characteristic functions,  
$$ \varphi_j (u) = \E\, e^{iuP_j^{*}} =  \int_{-\infty}^{\infty} e^{iux} \,\text{d} F_{p, j} (x) \,. $$
Correspondingly, we construct the Fourier transform of the kernel function $K$. Suppose we choose a proper kernel $K(u)$ such that the Fourier transform
$$ k(t) = \int_{-\infty}^{\infty} e^{itu} K(u) \, \text{d}u $$
is absolutely integrable. 
Then we derive the kernel density estimator in the form of the sample characteristic function. 
It follows that the kernel density estimator can be written as
\begin{equation*}
\begin{aligned}
{\hat f_P}^{(\xi,h)}(t) &=  \frac{1}{h\lfloor \xi m\rfloor} \sum_{j=1}^{\lfloor \xi m\rfloor} K\left( \frac{t-P^{*}_j}{h} \right) \\
&= \frac{1}{h} \int  K\left( \frac{t-x}{h} \right) \,\text{d}F_{p, \lfloor \xi m\rfloor} (x)
= \frac{1}{h} \int  K_h \left(  t-x  \right) \,\text{d}F_{p, s} (x) ,
\end{aligned}
\end{equation*}
where $K_h(t) = K(t/h)\,.$ 
Let $\mathscr{F}^\xi  \circ g$ denote the Fourier transform of a function $g$. 
The Fourier transform of ${f_P}_{( h)}(t)$ is then
\begin{equation*}
\begin{aligned}
\mathscr{F}^\xi  \circ {f_P}_{(h)}(t) & = \frac{1}{h} \mathscr{F}^\xi  \circ  K_h (t)  \cdot  \mathscr{F}^\xi  \circ F_{p, s} (t) \\
& =   \frac{1}{h} \int K_h (ut) e^{iut} \,\text{d}u  \cdot  \int  e^{iut} \,\text{d} F_{p, s}  (u) \\
&=  \frac{1}{h} \int K \left(\frac{ut}{h} \right) e^{iut} \,\text{d}u  \cdot  \int  e^{iut} \,\text{d} F_{p, s} (u) \\
&=  \int K (ut)  e^{ihut} \,\text{d}u  \cdot  \int  e^{iut} \,\text{d} F_{p, s} (u) \\
&=  k(ht) \varphi_s (t)\,,
\end{aligned}
\end{equation*}
and it follows that the estimator ${\hat f_P}^{(\xi,h)}(t) $ can be written as
\begin{equation}\label{eqn:FourierCh}
\begin{aligned}
{\hat f_P}^{(\xi,h)}(t)  = \mathscr{F}^\xi \circ \mathscr{F}^\xi \circ  {f_P}_{( h)}(t) 
&= \frac{1}{2 \pi} \int e^{-iut}  \mathscr{F}^\xi \circ  {f_P}_{( h)}(u) \,\text{d}u \\
&=  \frac{1}{2 \pi} \int e^{-iut} k(hu) \varphi_s (u) \,\text{d}u\,.
\end{aligned}
\end{equation}

In order to prove the convergence of ${f_P}_{ h}(t) \,,$ we consider
$$ \left |{\hat f_P}^{(\xi,h)}(t)   - f_{P}^{*} (t) \right |^2  = \left | {\hat f_P}^{(\xi,h)}(t)   -  
\E [ {\hat f_P}^{(\xi,h)}(t)  ] +  \E [ {\hat f_P}^{(\xi,h)}(t)  ] -  f_{P}^{*} (t)   \right |^2 \,,$$
and we utilise the fact that the kernel density estimate is asymptotically unbiased such that
$$\sup_t \left | \E [{\hat f_P}^{(\xi,h)}(t) ] -  f_{P}^{*} (t)   \right | \longrightarrow 0 ,  \quad s \to \infty\,,$$
so it is enough to prove 
\begin{equation}\label{eqn:LimSup}
\sup_t \left |{\hat f_P}^{(\xi,h)}(t)  -  \E [{\hat f_P}^{(\xi,h)}(t)  ]  \right | \longrightarrow 0 \,.
\end{equation}
Following (\ref{eqn:FourierCh}), we have
\begin{equation*}
\begin{aligned}
\left | {\hat f_P}^{(\xi,h)}(t)  -  \E [{\hat f_P}^{(\xi,h)}(t) ]  \right |  & \leq  \frac{1}{2 \pi} \int \left|  e^{iut}   k(hu) ( \varphi_s (u) - \E [\varphi_s (u) ] ) \right | \,\text{d}u  \\
& \leq   \frac{1}{2 \pi} \int   | k(hu) | \left | \varphi_s (u) - \E [\varphi_s (u) ] \right |\,\text{d}u  \,.
\end{aligned}
\end{equation*}

Consider the $L^2(\mathscr{P})$ norm of (\ref{eqn:LimSup}), it suffices to prove that
$$ \E^{\frac{1}{2}} \left[ \sup_t \left |{\hat f_P}^{(\xi,h)}(t)  -  \E [{\hat f_P}^{(\xi,h)}(t) ]  \right |^2 \right] \longrightarrow 0\,. $$
Notice that 
\begin{equation*}
\begin{aligned}
\E^{\frac{1}{2}} \left[ \sup_t \left | {\hat f_P}^{(\xi,h)}(t)   -  \E [{\hat f_P}^{(\xi,h)}(t) ]  \right |^2 \right] 
& \leq    \E^{\frac{1}{2}} \left[   \left |    \frac{1}{2 \pi} \int   | k(hu) | \left | \varphi_s (u) - \E [\varphi_s (u) ] \right | \,\text{d}u  \right |^2 \right]   \\
& \leq    \frac{1}{2 \pi} \int     | k(hu) | \, \E^{\frac{1}{2}} \left[ \left | \varphi_s (u) - \E [\varphi_s (u) ] \right |^2 \right]       \,\text{d}u \,,
\end{aligned}
\end{equation*}
which is a straightforward result from Minkowski's integral inequality and \newline
$\E^{\frac{1}{2}} \left[ \left | \varphi_s (u) - \E [\varphi_s (u) ] \right |^2 \right] $, which is the square root of the variance of $\varphi_s (u)$, is bounded by definition
$$ \text{Var}(\varphi_s (u) )  = \text{Var} \left( \frac{1}{s} \sum_{j=1}^{s} e^{iu P_j^{*}}\right) \leq \frac{1}{s} \E \left| e^{iu P_j^{*}} \right|^2 \leq \frac{1}{s}\,.$$
Therefore, 
\begin{equation*}
\E^{\frac{1}{2}} \left[ \sup_t \left |{\hat f_P}^{(\xi,h)}(t)  -  \E [ {\hat f_P}^{(\xi,h)}(t) ]  \right |^2 \right] 
 \leq  \frac{1}{\sqrt{s} h} \int |k(u)| \,\text{d}u \longrightarrow 0
\end{equation*}
if the bandwidth is chosen to satisfy
\begin{equation*}
s h^2 \to 0\,.
\end{equation*}
Thus,  (\ref{eqn:LimSup}) is proved and it follows that
\begin{equation*}
{f_P}_{(s, h)}(t)  \overset{p}{\longrightarrow} f_{P}^{*} (t), \quad s \to \infty\,,
\end{equation*}
and equivalently, as $m \to \infty.$
By  Lemma \ref{prop:Unimodal}, we conclude that the estimator of the mode with filtration converges to the theoretical mode $\vartheta^{\xi}$ of the filtered $p$-values, and therefore, converges to the mode $\vartheta$.
\end{proof}

\begin{remark}
In order to target the small $p$-values near 0 the transformation of the $p$-values to 
$-\log(p)$ is a good idea. 
The kernel density estimator based on the transformed values
is more reliable in practice, particularly when the mode of the alternative $p$-values is near zero, as is the case with Gaussian test statistics. 
\end{remark}



%
%
%
%


\subsection{Finite-sample control of the false discovery rate (FDR)}  
To capture the alternative $p$-values, We consider rejection regions equal to $\mathscr{R}({\hat\vartheta^{(\xi,h)}},\delta)=[\delta_l,\, \delta_u ]$ with $\delta_l=\hat\vartheta^{(\xi,h)}- \delta$ and $\delta_u =\hat\vartheta^{(\xi,h)}+\delta$, where the half-length is $\delta>0$.
The total number of rejections is 
$$R=R({\hat\vartheta^{(\xi,h)}},\delta) =\# \{  i:\,\hat\vartheta^{(\xi,h)}-\delta\leq p_i\leq \hat\vartheta^{(\xi,h)}+\delta\}\,.$$ 
The choice of $\delta$ is crucial and will be discussed next.
\subsubsection{Inference for FDR and pFDR}

The pFDR is investigated in   
\cite{storey2002direct},  \cite{storey2003positive} and \cite{storey2004strong} and considers $p$-values generated by a mixture model 
$$P_i | H_i \sim (1-H_i)\cdot U + H_i \cdot F_P \,,$$
where the indicators $H_i \sim \mathrm{Bernoulli}(\varepsilon)$ are independent random variables. The pFDR of a multiple test with rejection region $\mathscr{R}$ is defined as the conditional probability 
\begin{equation}
\label{Eqn:pFDRprop}
\mathrm{pFDR}(\mathscr{R} ) = \P (H_i = 0 | P_i \in \mathscr{R} ) = \frac{m\,(1-\varepsilon) \delta}{m\, \P (P_i \in \mathscr{R} ) } \,.
\end{equation} 
Here $\delta$ can be any half-width of the rejection interval.
%
The denominator of (\ref{Eqn:pFDRprop}) can be written as   
$$ \P (P_i\in \mathscr{R} ) = \P (P_i\in \mathscr{R} \,|\,R>0) \P (R>0)\,,  $$
where $\P (P_i \in \mathscr{R} \,|\,R>0)$ can be estimated using the observed value of $R$, that is, 
$$\hat{\P} (P_i \in \mathscr{R} \,|\,R>0)= \frac{R\vee 1}{m}\,.  $$

The probability $\P (R>0)$ depends on the distribution of the test statistics and the multiple testing procedure.  
For our rejection region $\mathscr{R}({\hat\vartheta^{(\xi,h)},\delta})$ we assume that the $p$-values outside of the rejection region are uniformly distributed. 
This leads to an estimate of $\P (R>0)  =1-\P (R=0) \leq 1-(1-\delta)^m$.   

The pFDR is thus estimated as
\begin{equation*}
 \widehat{\text{pFDR}}(\mathscr{R}({\hat\vartheta^{(\xi,h)}},\delta)) = \frac{(1-\hat \varepsilon) m \delta}{ \{R\vee 1\} (1-(1-\delta)^m)}\,, 
\end{equation*}
and the FDR, without conditioning on $R>0$, is estimated by
\begin{equation} 
\label{Eqn:FDRhat}
\widehat{\text{FDR}}(\mathscr{R}(\hat\vartheta^{(\xi,h)},\delta)) = \frac{ (1-\hat \varepsilon) m  \delta}{ \{R\vee 1\}}\,. 
\end{equation}

We next consider the estimate of $\varepsilon$. 

\begin{enumerate}
\item 
The simplest idea is to avoid estimating $\varepsilon$ by using $1-\varepsilon \leq 1$ in which case the numerator of  (\ref{Eqn:FDRhat}) is estimated by $m \delta$. Although this estimator 
$$  \widehat{\mathrm{FDR}}({\hat\vartheta^{(\xi,h)}},\delta) = \frac{ m  \delta}{ \{R({\hat\vartheta^{(\xi,h)}},\delta)\vee 1\}}   $$
is widely used, and does provide a bound of the FDR control, we seek to have a more precise and more readily interpretable estimate of $\varepsilon$.

\item  
Based on our filtering procedure, we can select an optimal filtering parameter $\xi$ and use it as an estimator of $\varepsilon$. 
Suppose the filtering procedure starts with an initial $\xi_0=1/2$, estimates the mode and then keeps decreasing the $\xi$-values until   
the estimate  $\hat \vartheta$ stabilizes in the sense that the change in the sample mode is smaller than $C$ when changing the $\xi$-value. The second-to-last value of $\xi$ can be then be taken as $\hat \varepsilon$.     
The estimated FDR is 
 \begin{equation*} 
\widehat{\text{FDR}}(\mathscr{R}({\hat\vartheta^{(\xi,h)}},\delta)) = \frac{ (1- \xi) m  \delta} {R(\hat\vartheta^{(\xi,h)}, \delta)\vee 1}\,. 
\end{equation*}

\item   
\cite{storey2002direct} introduced the estimator based on a fixed rejection region pre-specified by a tuning parameter  $\lambda \in (0,1)$, which is given by   
\begin{equation*} 
1-\hat{\varepsilon} = \frac{W (\lambda)}{(1-\lambda)m}\,,
\end{equation*}
where $W (\lambda)= \# \{  p_i > \lambda  \}$ is the number of accepted nulls with the length of acceptance region $1-\lambda.$  
The parameter $\lambda$ is selected by minimising the mean square error of the FDR estimator. 
This estimator was later  discussed in 
 \cite{storey2003positive},  \cite{storey2004strong},  
\cite{genovese2004stochastic}, \cite{benjamini2006adaptive} and other related works. 
Benjamini used this estimator  
and replaced $m$ by $\hat m_0 = m(1-\hat \varepsilon)$ in the denominator of the linear step-up threshold, that is, to reject the nulls for the  $p$-values $p_{(i)} \leq  i \alpha / \hat m_0. $ 
Others used this estimator in the inference and control of FDR. 

There are also the methods given by 
\cite{benjamini2000adaptive},  
\cite{meinshausen2006estimating},  
\cite{cai2007estimation},  
\cite{jin2007estimating}
that contribute to the estimation of $\varepsilon$, but the estimators are too complicated or not applicable to our study. 

\end{enumerate}

We use the estimator of $\varepsilon$ defined by 
\begin{equation} 
\label{Eqn:EstimateEpsilon}
1-\hat{\varepsilon} = \frac{W_{\hat\vartheta^{(\xi,h)}}(\xi)}{(1-\xi)m}\,,
\end{equation}
which is analogous to Storey's method, with $W_{\hat\vartheta^{(\xi,h)}}(\xi) = \# \left\{ |p_i - \hat\vartheta^{(\xi,h)}| > \xi /2 \right\}$.   
Note that 
\begin{multline*} 
\E (\hat{\varepsilon}) 
= 1- \frac{\E (\sum_{i=1}^m\textbf{1}\{| p_i-\hat\vartheta^{(\xi,h)} |>\xi / 2 \})  }{(1-\xi)m} 
\leq 1- \frac{\E (\sum_{H_i=0} \textbf{1}\{|p_i-\hat\vartheta^{(\xi,h)}|>\xi / 2 \})  }{(1-\xi)m} \\
= 1- \frac{(1-\xi)m_0}{(1-\xi)m} = \varepsilon\,. 
\end{multline*}
Since the simplest structure we assume is the two-point mixture model,    
we can expect that the $p$-values outside the estimated region $\mathscr{R}_{\vartheta, \xi}$ are dominated by the true nulls with  frequency $(1-\varepsilon)(1-\xi)$. 
In other words,  $W_{\hat\vartheta^{(\xi,h)}}(\xi)$ is roughly $(1-\xi)(1-\varepsilon)m,$  
since  the  $p$-values from the nulls are assumed to be uniformly distributed outside the rejection region of which the length $1-\delta$ is replaced by $1-\xi\,.$  
Thus, $W_{\hat\vartheta^{(\xi,h)}}(\xi)/(1-\xi)$ is  analogous to $m-R$, and is therefore utilised to estimate $m-\varepsilon m$. 
\begin{remark}
Although we desire to find an accurate estimate of the fraction of the effects,  a lower bound of $\varepsilon$ suffices, and the parameter $\xi$ needs not be an estimator of $\varepsilon$. 
The estimator (\ref{Eqn:EstimateEpsilon}) is slightly biased, and can be used as a lower bound of $\varepsilon$.  
In reality it is acceptable to claim that the proportion of true alternatives is no less than the declared frequency $\hat \varepsilon$ in  expectation.  
\end{remark}

The final estimator of pFDR is
\begin{equation}
\label{Eqn:pFDRhat}
 \widehat{\text{pFDR}}(\vartheta, \delta) = \frac{W_{\hat\vartheta^{(\xi,h)}}(\xi)\delta}{(1-\xi)\{R\vee 1\} (1-(1-\delta)^m)}\,, 
\end{equation}
while the FDR (see \ref{Eqn:FDRhat} and \ref{Eqn:EstimateEpsilon}) is estimated by
\begin{equation} 
\label{Eqn:FDRhatbis}
\widehat{\text{FDR}}(\vartheta, \delta) = \frac{W_{\hat\vartheta^{(\xi,h)}}(\xi)\delta}{(1-\xi)\{R(\vartheta, \delta)\vee 1\}}\,. 
\end{equation}
Based on these estimators we are able to implement the following data-dependent algorithm to detect and locate the alternatives. 
\begin{enumerate}
\item  Compute the $p$-values $p^m = \left\{ p_1, \ldots, p_m \right\}$ and order them non-decreasingly, i.e.  $p_{(1)} \leq p_{(2)} \leq \cdots \leq p_{(m)}.$ 
\item  Apply the fixed length filter $\mathcal{T}^\text{fixed}$ with a parameter $0<\xi<1$ to obtain the filtered sequence $\mathscr{F}^\xi$. 
\item  Estimate the density of $\mathscr{F}^\xi$ and the mode $\hat\vartheta=\hat\vartheta^{(\xi,h)}$. 
\item  Given $i=1, \ldots, m$ and a significance level $\alpha,\,$  reject $H_{0,(i)}$ if $ | p_{(i)} - \hat \vartheta | \leq   | p_{(\tau)}- \hat \vartheta |$, where $\tau =  \max \left\{\tau: \widehat{\text{FDR}}(\hat \vartheta , p_{(\tau)} - \hat \vartheta) \leq \alpha \right\}.$ 
\end{enumerate}
The control of FDR is based on the estimator   (\ref{Eqn:FDRhatbis}). 
The estimator $\hat \delta = p_{(\tau)}- \hat \vartheta$ is equal the
largest half-width of the rejection region
$\mathscr{R}(\hat\vartheta^{(\xi,h)},\hat\delta)$ subject to the
control of $\widehat{\mathrm{FDR}}$.

\subsubsection{Data-dependent control of FDR for finite sample}

The following theorem gives some understanding of the estimator of the FDR. 

\begin{theorem}
Based on the rejection region  $\mathscr{R}_{\hat\vartheta^{(\xi,h)}} $ and 
The estimator of FDR given by (\ref{Eqn:FDRhat}) satisfies 
\begin{equation*}
\mathbf{E}(\widehat{\text{FDR}}(\vartheta, \delta)) \geq \mathrm{FDR}(\vartheta, \delta) 
\end{equation*} 
for any valid $(\vartheta, \delta)$.
\end{theorem}

\begin{proof}

We take the difference 
\begin{multline*}
 \E (\widehat{\text{FDR}}(\vartheta, \delta)) - \text{FDR}(\vartheta, \delta)   
=  \E \,\left[ \frac{\delta  W_{\vartheta}(\xi)/(1-\xi)}{\{R(\vartheta, \delta) \vee 1 \}  }\right]
  - \E \,\left[ \frac{ V(\vartheta, \delta)}{\{R(\vartheta, \delta) \vee 1 \} }\right] \\
=   \E \,\left[ \frac{\delta W_{\vartheta}(\xi )/(1-\xi ) - V(\vartheta, \delta) }{\{R(\vartheta, \delta) \vee 1 \}   }\right]  
\geq   \E \,\left[ \frac{\delta W_{\vartheta}(\xi )/(1-\xi ) - V(\vartheta, \delta) }{ R(\vartheta, \delta)   }  \textbf{1} \left\{ R(\vartheta, \delta) >0  \right\}\right] . 
\end{multline*}
Recalling that 
$$ R(\vartheta, \delta) = S(\vartheta, \delta)  + V(\vartheta, \delta), $$
we condition on $S(\vartheta, \delta) $ and tackle the $V(\vartheta, \delta)$ in both the numerator and the denominator.  
We obtain that the last equation above equals 
\begin{equation*}
\begin{aligned}
& \E \,\left[ \frac{\delta W_{\vartheta}(\xi )/(1-\xi ) - V(\vartheta, \delta) }{ S(\vartheta, \delta)  + V(\vartheta, \delta)   }  \textbf{1} \left\{ R(\vartheta, \delta) >0  \right\}\right]   \\
=  & \E \,\left[  \E \left[ \frac{\delta W_{\vartheta}(\xi )/(1-\xi ) - V(\vartheta, \delta) }{ S(\vartheta, \delta)  + V(\vartheta, \delta)   }  \textbf{1} \left\{ R(\vartheta, \delta) >0  \right\} \Big |   S(\vartheta, \delta)   \right]  \right]  \\
\geq  & \E \,\left[  \frac{ \E \left[ \left( \delta W_{\vartheta}(\xi )/(1-\xi ) - V(\vartheta, \delta)  \right) \textbf{1} \left\{ R(\vartheta, \delta) >0  \right\}  \big |  S(\vartheta, \delta)   \right] }
{ \E \left[  ( S(\vartheta, \delta)  + V(\vartheta, \delta) ) \textbf{1} \left\{ R(\vartheta, \delta) >0  \right\}   \big | S(\vartheta, \delta)   \right]    }  \right] \,,
\end{aligned}
\end{equation*}
with the last inequality obtained by Jensen's inequality on
$V(\vartheta, \delta)$,\, given the fact that
 $$ \frac{W -V}{S + V} =\frac{W+S}{V+S} -1 $$ 
 is a convex function of $V$ with $W+S > 0$.   
Since
\begin{equation*} 
\E \left[   \delta W_{\vartheta}(\xi )/(1-\xi ) -  V(\vartheta,
  \delta)      \right]  \geq  \delta m(1-\varepsilon)(1-\xi )/(1-\xi )
-  m(1-\varepsilon) \delta  =  0, 
\end{equation*}
we conclude that
$$  \E (\widehat{\mathrm{FDR}}(\vartheta, \delta)) \geq \mathrm{FDR}(\vartheta, \delta) \,. $$
\end{proof}

Similarly, we obtain that
\begin{equation*}
 \E (\widehat{\mathrm{pFDR}}(\vartheta, \delta)) \geq \mathrm{pFDR}(\vartheta, \delta) \,, 
\end{equation*}
with our estimator having $1-(1-\delta)^m \geq \P (R>0)  $ in the denominator of (\ref{Eqn:pFDRhat}). 

Following this theorem we can get control of the true FDR by limiting
the estimated $\widehat{\text{FDR}}(\hat\vartheta, \delta)$ below a
desired level.  Our rejection region
$\mathscr{R}(\hat\vartheta, \delta)$ is nested, and the monotonicity
of power is guaranteed.

\begin{prop}[Monotonicity of power]
  For fixed center $\hat\vartheta\,,$ the decision rule defined by the
  rejection region $\mathscr{R}(\hat\vartheta, \delta)$ has monotone
  power in a sense that
$$ \beta_{\hat(\vartheta, \delta)} \geq \beta_{\hat(\vartheta, \delta{'})} \quad \text{for any} \quad  0 < \delta \leq \delta{'} \leq  2\hat\vartheta\,. $$
\end{prop}

\begin{remark}
  \cite{storey2003positive} also considered the asymptotic control of
  the FDR and pFDR with $\varepsilon_m = \varepsilon$ fixed.  We are
  not interested in this parametrisation since the number of
  significant components can be moderately large if it is proportional
  to $m$.
\end{remark}

\subsection{Numerical results}

With $m=1000$ hypotheses, we applied the proposed filtration algorithm
to Cauchy mixtures with $\varepsilon=0.05, 0.10, 0.15, 0.20, 0.25,$
and $\mu= 6, 8, 10, 12, 14, 16, 18, 20$.  For each configuration we
ran $N=200$ replications and get the sample value of the parameters
and the true and false discoveries.

\begin{table}
\caption{\label{Tab:SimCauchyFDR} Simulation results with shifted Cauchy test statistics. The entries show the average values across 1000 replicas for the case of $\varepsilon=$ 0.15. The empirical FDR is fixed at 0.10. The entry for $\widehat{\mathrm{FDR}}$ is the average for the estimated FDR.}
\centering
\begin{tabular}{c|c|c|c|c|c|c|c|c}
$\varepsilon = 0.15$           &	$\mu=6$	& $\mu=8$	&$\mu=10$	&$\mu=12$	&$\mu=14$	&$\mu=16$	&$\mu=18$	&$\mu=20$ \\  
$\vartheta$  			  & 0.05121 & 0.03899 & 0.03142 & 0.02628  & 0.02258  & 0.01979  & 0.01761  & 0.01586    \\   
$\hat{\vartheta}$		  & 0.05193 & 0.03935 & 0.03152  & 0.02636  & 0.02264  & 0.01984  & 0.01763  & 0.01587   \\  
$\hat{\delta}$ 			  & 0.01573  & 0.01521 & 0.01412  & 0.0139  & 0.01317 & 0.01256  & 0.01241    & 0.01288   \\	 
$\widehat{\mathrm{FDR}}$ & 0.08577  & 0.09389 & 0.08816  & 0.08660  & 0.08699  & 0.08493 & 0.08364  & 0.08377     \\	 
FDR  				  & 0.08547 & 0.09238  & 0.08920  & 0.08258 & 0.08364  & 0.07820  & 0.07717 & 0.07451   \\	
$\mathrm{TP}/ (\varepsilon m)$		 & 0.4924  & 0.4967 & 0.5994 & 0.7506 & 0.8175 & 0.8702 &  0.9006 & 0.9219  \\  
\end{tabular}
\end{table}

\begin{figure}
\centering
\includegraphics[width=0.6\textwidth]{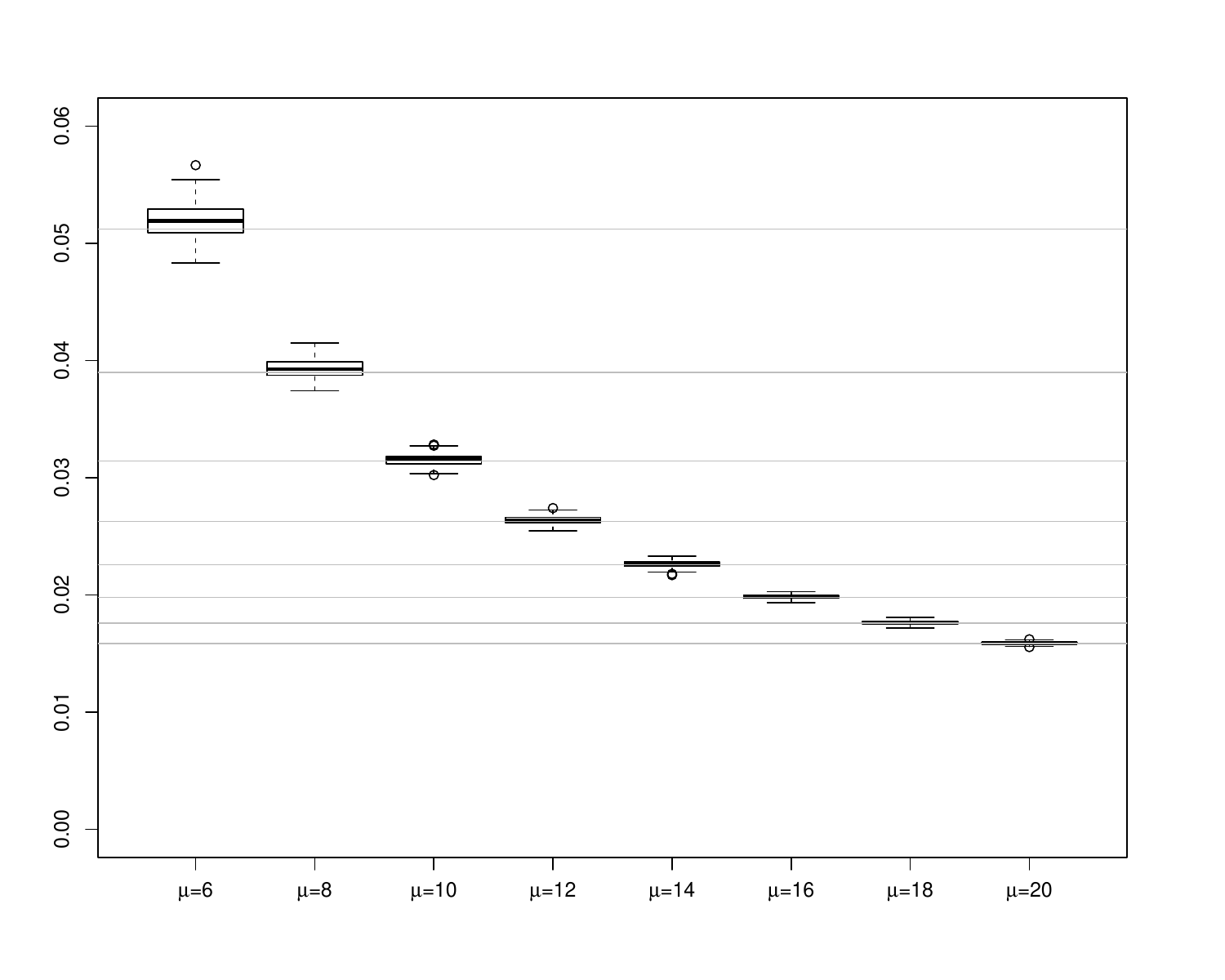}
\caption{\label{Fig:BoxplotModeEst} The filtering estimate of the mode of the alternative $p$-values. The estimate is based on the kernel density estimator with the Gaussian kernel and bandwidth tuned to the remaining $p$-values. See also Eq. (\ref{Eq:mode})}

\end{figure}

The estimates of the alternative mode $\hat \vartheta$ are shown in
Figure \ref{Fig:BoxplotModeEst}. The mode estimates decrease as the
shift $\mu$ increases, that is, smaller $p$-values become indicators
for true alternatives.  A pre-specified level $\alpha=0.1$ is utilised
to control the data-dependent estimator
$\widehat{\mathrm{FDR}}(\vartheta, \delta) \leq \alpha$ given by
(\ref{Eqn:FDRhat}).  We choose the rejection region
$\mathscr{R}({\hat \vartheta, \hat \delta})$ with the maximal length
$\hat \delta$ subject to the control of
$\widehat{\mathrm{FDR}} \leq \alpha$.  The average FDR is shown in the
table by $\widehat{\mathrm{FDR}}$.  The true value of FDR computed
from the sample is different from the estimator
$\widehat{\mathrm{FDR}}(\vartheta, \delta)$, which is influenced by
the tuning parameter $\xi$ as we propose in the filtering procedure.
With the peak of the $p$-value getting narrow, the rejection region
contains more true alternatives.

\section{Discussion}

\subsection{Positive FDR, local FDR and empirical Bayes} 

In this section we compare our procedures to related ideas and concepts described in 
\cite{efron2001empirical},  \cite{efron2002empirical},  \cite{storey2002direct},  \cite{storey2003positive},
and  
 \cite{cai2017optimal}. 

\subsubsection*{pFDR}

Storey's pFDR is also referred to as the posterior FDR, because it
applies the Bayes formula to the FDR using the independent Bernoulli
model with a common prior probability for $H_i=1$.  Recall the formula
for the pFDR of the rejection region $\mathscr{R}$
\begin{equation*}
  \mathrm{pFDR}(\mathscr{R}  ) =  \P ( H_i=0   | P_i \in \mathscr{R}  ) 
  = \frac{(1-\varepsilon) \P(P_i \in \mathscr{R}  | H_i=0 )
  }{(1-\varepsilon) \P(P_i \in \mathscr{R}  | H_i=0 )  +  \varepsilon
    \P(P_i \in \mathscr{R}  | H_i=1 ) } \,. 
\end{equation*}
This is equivalent to our control of the operating characteristics TPR/FPR, 
However, Storey and other authors of the related work only discuss the
case when
$ \P(P_i \in \mathscr{R} | H_i=1 ) / \P(P_i \in \mathscr{R} | H_i=0 )$
is decreasing, which is the same condition we mentioned for detecting
light-tailed alternatives.
These papers also limit discussion to the asymptotic cases under the assumption that 
$$ \sum_{i=1}^m (1-H_i) /m \longrightarrow \pi_0,  \quad m \to \infty, $$
while we consider the asymptotic framework with 
$$  \sum_{i=1}^m (1-H_i) /m = 1-\varepsilon_m = 1- m^{-\gamma}  \longrightarrow 1,  \quad m \to \infty. $$
Although the numbers of the true nulls and alternatives both tend to
infinity, the ratio will be difficult to detect, which also motivated
us to investigate the asymptotically detectable region, that is,
detectable clustering in the limit.

\subsubsection*{Local FDR}

The \textit{local false discovery rate} was originally developed for the $z$-values
and uses results from empirical Bayes inference.  The $z$-values have
densities $f_0(z)$ under the nulls and $f_1(z)$ under the alternatives
with
$f_0(z) = \varphi(z) = \frac{1}{\sqrt{2\pi}} \mathrm{e}^{-z^2/2}$.
With the fixed prior probabilities $\pi_0$ and $\pi_1$, the density of
the observed $z$-values is $f(z) = \pi_0 f_0(z) + \pi_1 f_1(z)$. The
\textit{local Bayes false discovery rate} is then defined as
 \begin{equation*}
 \mathrm{Lfdr}(z) = \P(\text{null } | \text{ test statistic } z) = \frac{ \pi_0 f_0(z)}{ \pi_0 f_0(z)+ \pi_1 f_1(z)} \,,
 \end{equation*}
where the densities in the numerator and the denominator need to be estimated. 

In our work, we consider the distribution of the $p$-values, and we
maximise the local ratio TPR(t)/FPR(t) to get the significance center
such that a large number of true positives are discovered subject to a
small increment of the false positives.  We can equivalently define
the local FDR for the $p$-values as
$$\mathrm{Lfdr}_p(t)  = \frac{1-\varepsilon }{1-\varepsilon +\varepsilon\,f_P(t)} .$$ 
Efron's density estimation makes use of the normal distribution of the
$z$-values, whereas we rely on the uniform distribution of the
$p$-values from the null hypotheses.  

Maximising the local ratio of TPR/FPR  is equivalent to minimising the
local FDR, taking the Lfdr(t) as a point-wise threshold sequence
defined for the $p$-values,  and in addition, equivalent to minimising
the pFDR as well.  
We are particularly interested in looking for the most informative
region of the $p$-values without a pre-determined rejection rule.   
Our method is adaptive and data-dependent, and is also interpretable.

\subsubsection{Screening for high-throughput data}

A filtering method similar to ours appeared in \cite{cai2017optimal}
for a different purpose.  In high-dimensional multiple testing, one of
the main issues is to reduce the dimension according to the capacity
of the experiments.  They discussed a screening approach applied to
high-throughput applications, leading to a multi-stage procedure.
Their selection rule is defined by the indicator
$$ \delta_i = \textbf{1} \{ \hat T(Z_i) \leq  t_i\} \,, $$
where $\hat T(Z_i) $ is an estimator of the local FDR and $t_i$ is a
critical value.  The observations with $ \delta_i =1$ are retained and
their paper derives the conditions necessary for a valid screening
procedure.  They use classic kernel density estimation to estimate the
densities and the effectiveness of their estimate Lfdr($z$) depends on
the distribution of the test statistics.

\subsection{Multi-mode estimation and rejection sets}

Consider the mixture model 
\begin{equation*}
f_1(x) = \sum_{j=1}^k \pi_j f_0(x - \mu_j), 
\end{equation*}
of which the proportions $\pi_i$'s  and the shifts $\mu_i$'s  are unknown and not identical.    
Following the idea of the two-point mixture model, we propose the rejection sets
\begin{equation*}
\mathscr{R} =  \left\{ \bigcup_{i=1}^{k}  \mathscr{R}^{(i)} \right\}
\end{equation*}
where
$ \mathscr{R}^{(i)} =  \mathscr{R}_{\vartheta_i, \delta_i} $ is the $i$-th rejection interval.  

This problem of detecting clustered alternative components can be converted to a problem of change point detection, as is analysed by 
\cite{siegmund2011false},   \cite{zhang2010detecting},  \cite{cao2015changepoint} et al. 

\begin{figure}
\centering
\includegraphics[width=0.6\textwidth]{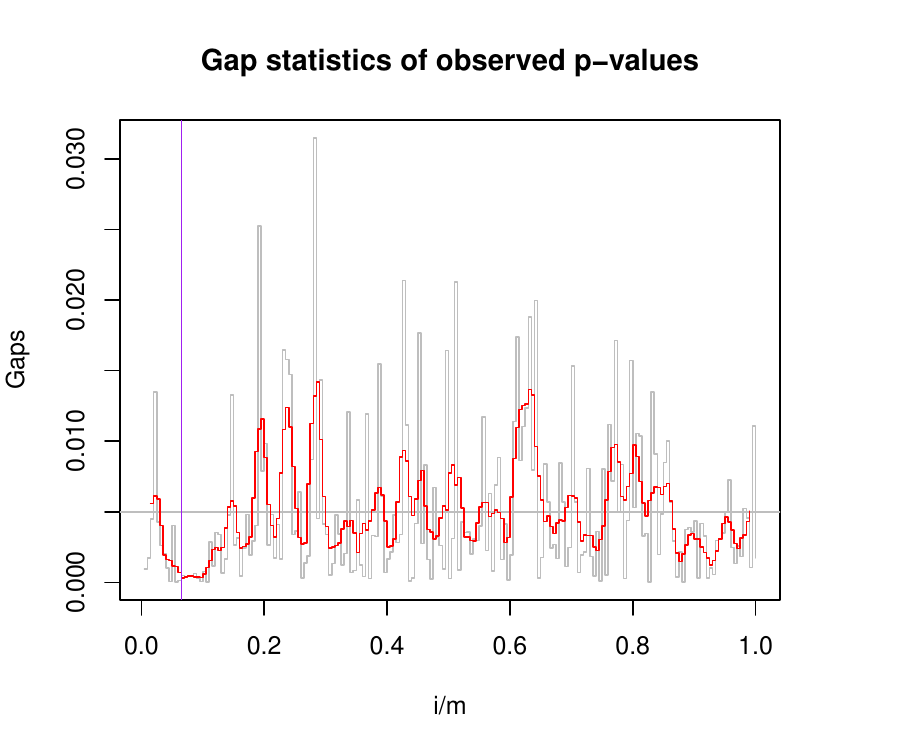}
\caption{\label{Fig:SmoothedGaps} The smoothed gaps of the $p$-values  from the Cauchy mixture model.}

\end{figure}

Instead of detecting the change point in the sequence of $p$-values, we propose an approach to detect the rejection centers based on smoothed gap statistics. 

Define a \textit{smoothed} version of observed $p$-values
\begin{equation*}
 p_{j}^{\dagger} = \frac{1}{j} \sum_{i=1}^{j} p_{(i)},
 \end{equation*}
for $j=1, \ldots, m,$ and let $p_0^{\dagger} = p_{(0)} = 0.$ For $j=1, \ldots, m$ we define the weighted gap statistic
\begin{equation*}
\begin{aligned}
G_j^{\dagger} &=  p_{j}^{\dagger} -   p_{j-1}^{\dagger}  \\
 &= \frac{(j-1)(p_{(j)}-p_{(j-1)}) + (j-2)(p_{(j-1)}-p_{(j-2)}) + \cdots + (p_{(2)}-p_{(1)})}{j(j-1)}  \\
 &=  \frac{(j-1) G_{j-1} + (j-2) G_{j-2} + \cdots + G_1 }{ j(j-1)} \,,
\end{aligned}
\end{equation*}
which can be re-written as a weighted sum of the original gap statistics $G_j = p_{(j)} - p_{(j-1)}$.  
We give a larger weight to $G_i$  as it is closer to  $G_j^{\dagger},$
which means that the weighted sum of gap statistics capture more
precisely the local properties of the $p$-values.  

When the observations are i.i.d. from the null distribution, the $p$-values are uniformly distributed, 
and it follows that $G_i = p_{(i)}-p_{(i-1)} \sim \text{Beta}(1, m)$,
with   expectation $\E (p_{(j)}-p_{(j-1)})=1/(1+m)$.  
The weighted gaps have a Beta distribution with $\E (G^{\dagger}_{j})
= \E (p^{\dagger}_{j} - p^{\dagger}_{j-1})=1/(2+2m)$.  
Therefore, it is reasonable to compare the weighted gaps to $1/(2+2m)$
and find the region where the cluster of  alternative $p$-values
occurs, if any.

\begin{figure}
\centering
\includegraphics[width=0.6\textwidth]{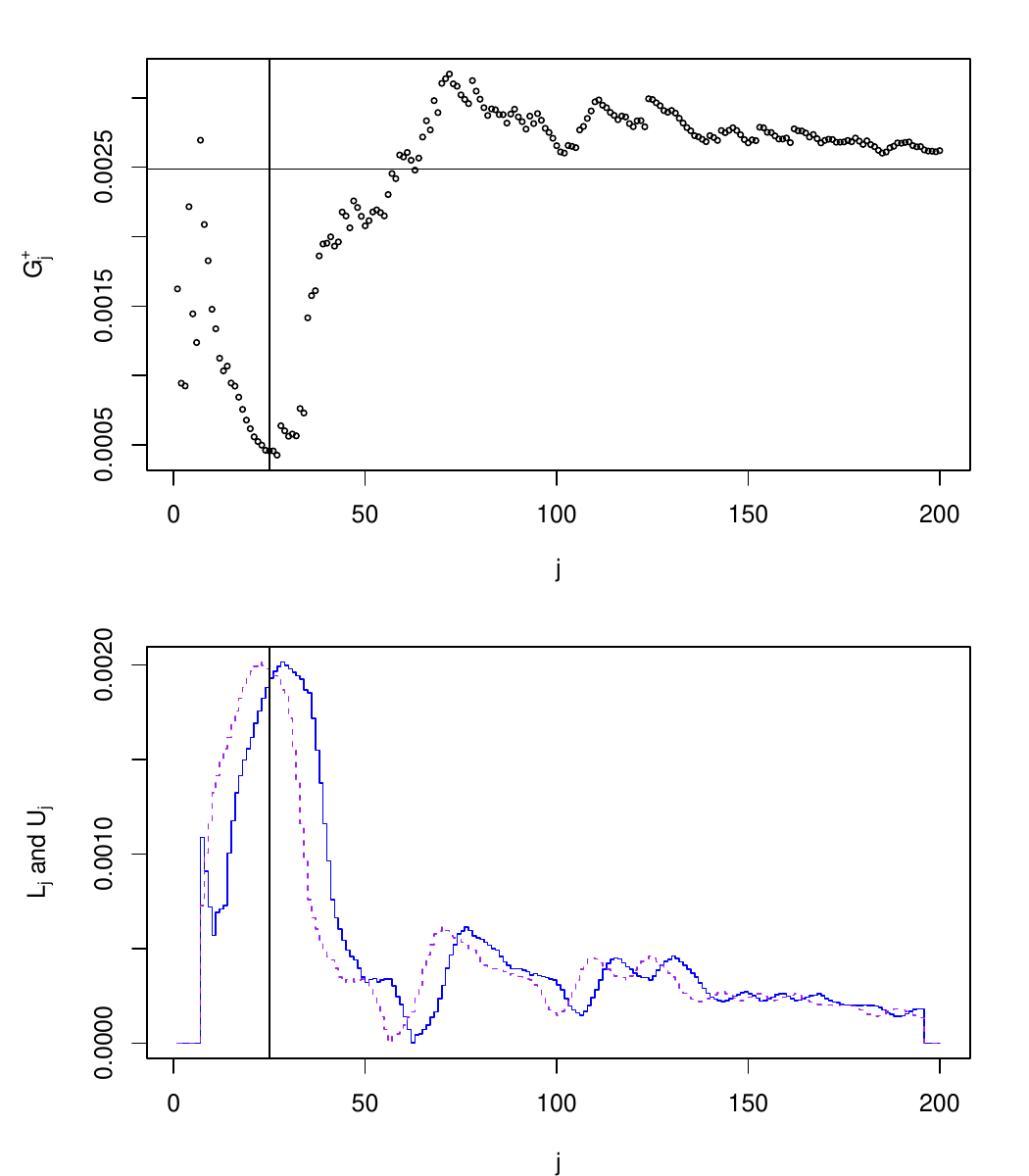}
\caption{\label{Fig:WeightedGapsLocalDiscrepancy} The weighted gaps and the lower and upper local discrepancies}
\end{figure}

For two-point mixture models, one can intuitively take the $p$-value that minimises the weighted gap, denoted by $\hat{p}_c$,  to be the center of the cluster from the alternatives. 
The change point of $G_j^{\dagger}$ gives a plausible estimation of the significance center. 
Formally, define the local discrepancies 
\begin{equation*}
  L_j = \left | \frac{1}{k}\sum_{i=j-k}^{j-1}
    G_j^{\dagger}-\frac{1}{2(m+1)} \right | ,  \quad                         
  U_j = \left | \frac{1}{k}\sum_{i=j}^{j+k-1} G_j^{\dagger}-\frac{1}{2(m+1)} \right | ,
\end{equation*}  
where $L$ stands for ``lower'' and $U$ stands for ``upper''.   
We want $L_j$ and $U_j$ to be sensitive to the change in the distribution of the $p$-value gaps.    

\subsection*{Conclusions} 

The multiple testing literature takes it as a given, that the true
alternatives have very small $p$-values. This assumption is wrong in
the case of test statistics with long-tailed laws. More general
rejection regions can be adaptively selected based on the observed
$p$-values. We present such a robust multiple testing procedure and
examine its properties. Our approach uses a filter that enlarges the
proportion of true alternatives among the filtered $p$-values and then
estimates a center for the rejection region by estimating the location
of the mode of the filter results. An interval around this center is
chosen in order to keep control over the FDR. In some instances, it
may be necessary to consider multiple modes in the density of the
alternative $p$-values. It would be straightforward to generalize the
methods discussed here to this case. 

The mode estimator  is thus utilised  as the mid-point of the central
peak of the $p$-values from the alternatives, which in our definition,
serves as the significance center of the rejection region
$\mathscr{R}.$   
Unlike for the Gaussian test procedures, we define the rejection
region $\mathscr{R}({\hat\vartheta, \delta})$ centered at the mode
$\hat\vartheta$ and of length $\delta$.   
The center $\hat\vartheta$ is estimated by a kernel density estimation
applied to the filtered $p$-values, and the length  $\delta$ is chosen
by data-dependent control of the FDR.  
We proved that the expected value of the estimator of FDR provides a
good upper bound of the true value of the estimated FDR, such that
this data-dependent control functions well.    
In this procedure we do not propose an estimate of $\delta$. 
An optimal $\hat{\delta}$ is chosen to achieve the maximal  power with
the  estimated $\mathscr{R}({\hat\vartheta,\hat\delta})$  bounded by
$\alpha. $ 

\bibliography{bibliography}
\end{document}